\documentclass[12pt]{article}

\usepackage{graphicx} 
\usepackage{cite}
\usepackage{amsmath}  
\usepackage{amssymb}  
\usepackage{amsfonts}
\usepackage{amsthm}
\usepackage[english]{babel}\usepackage[latin1]{inputenc}
\usepackage{array}
\usepackage[noend,algoruled]{algorithm2e} 

\usepackage{pgfplots}
\usetikzlibrary{patterns}
\usepackage{tikz}
\usetikzlibrary{arrows,matrix,positioning} 

\renewcommand{\Indp}{\advance\skiprule by 1em\advance\skiptext by -1em\advance\leftskip by 1em}
\renewcommand{\Indm}{\advance\skiprule by -1em\advance\skiptext by 1em\advance\leftskip by -1em}
\SetCommentSty{textit}
\SetKwComment{Comment}{$\quad$ // }{}
\SetKw{Echec}{fail}
\SetKw{KwAnd}{and}
\SetKwRepeat{DoWhile}{do}{while}

\newtheorem{deff}{Definition}

\newtheorem{proposition}[deff]{Proposition}

\newtheorem{example}{Example}
\newtheorem{remark}{Remark}

\newcommand{\CC}{\mathcal{C}}
\newcommand{\LC}{\mathcal{L}}
\newcommand{\CCpi}{\CC_{\pi}}

\newcommand{\G}{\mathcal{G}}
\newcommand{\Gtilde}{\tilde {\mathcal{G}}}

\newcommand{\eqdef}{\stackrel{\text{def}}{=}}

\newcommand{\xb}{\mbox{\boldmath $x$}}

\newcommand{\yb}{\mbox{\boldmath $y$}}
\newcommand{\zb}{\mbox{\boldmath $z$}}

\newcommand{\ec}{\mathcal{E}}
\newcommand{\smax}{s_{\text{max}}}

\newcommand{\Gt}[1]{\tilde{{\mathcal{G}}}(#1)}

\author{Audrey Tixier\\
{Inria,   Le Chesnay BP 105, 78153 C\'edex} }

\title{Blind identification of an unknown interleaved convolutional code}
\date{}

\begin{document}

\maketitle

\section{Introduction}
\par{\em Blind identification of an unknown code from the observation of noisy codewords.}
We address here a problem related to cryptanalysis and data security where an observer 
wants to extract information from a noisy data stream where the error correcting code which is used is unknown. 
Basically an observer has here several noisy codewords originating from an unknown code and wants to recover
the unknown code and decode it in order to recover the whole information contained in these codewords.
Generally this problem is solved by making assumptions on the code which is used in this scenario (convolutional code,
LDPC code, turbo-code, concatenated code, etc.). This is called the code reconstruction problem or blind identification
of a code problem
in the literature.
This problem arises for instance in a non-cooperative
context where observing a binary sequence originating from an
unknown communication system naturally leads to such a problem, for more details see the introduction of \cite{MGB12a}.
It also arises in the design of cognitive receivers which are able to cope with a great variety of error correcting codes \cite{MGB12a} or
in the study of DNA sequences  when looking for possible error correcting codes in the genetic code \cite{Ros06}.
 This problem has a long history: it has
been addressed for a variety of
codes, linear codes \cite{Val01,BSH06,CT08,CF09},
cyclic codes \cite{ZHSY, WYY10,WHHQ10,YVK14}, 
LDPC codes \cite{CT08,CF09},
convolutional codes  \cite{Ric95,Fil97,Fil00,LSLZ04,DH07,WHZ07,CF09a,CS09,MGB09a,Mar09,MGB12a,MGB12b,ZMGR11,ZGMRR12, ZHSZ14,SZHLZ14, BT14},
turbo-codes \cite{Bar05,Bar07,CS10,CFT10,NAF11,DWJ12b, MGB09b, GMB08, YLP14, LG13, TTS14}, 
BCH codes \cite{LPLS12, WYY11,ZHLSZ13,WL13}, 
Reed-Solomon codes \cite{XWH13,LLWC13,ZWJ13} 
and Reed-Muller codes \cite{KLPSS11}.

We focus here on the problem of reconstructing an unknown code when an interleaver is applied after encoding. 
Recall that this is a commonly used technique to correct burst errors since it provides some sort of time diversity in the coded 
sequence. For instance when the code is a convolutional code, this allows to spread the burst errors
in remote locations and convolutional decoding performs much better. This problem has been already addressed in a series of papers  \cite{LLG09b, GLLL13, JLLG12a, JLLG12b,JYLC11}. All these papers assume that the interleaver is 
structured (a convolutional interleaver \cite{GLLL13, JLLG12a, JLLG12b, LLG09b} or an helical scan interleaver \cite{JYLC11}). It should be said here that the methods used in these papers make heavily use of the 
particular structure of the interleaved that is considered and some of  these methods are highly sensitive to noise
\cite{JLLG12b,GLLL13} or do not apply when there is noise \cite{LLG09b,JLLG12a}.

In this paper, we will focus on the case where the code is a convolutional code and where the interleaver is a block interleaver. 
This pair is used in several standards, for example, in 802.11n, 802.16e, 802.22, GMR-1 and Wimedia. In practice the block interleaver is structured but this structure is not always the same, so to reconstruct the interleaver in all cases we will assume that the interleaver has no particular structure, it is chosen randomly among all possible permutations.
 Our problem can be formalized as follows.

\par{\em Blind identification of an unknown interleaved convolutional code: statement of the problem,  hypotheses and notations.}
The encoding process which is studied in this paper is described in Figure \ref{fig:schema_codage} and consists in taking an information word of length $m k$ and feeding it into an $(n,k)$ convolutional encoder to produce a codeword $\xb$ of length $N = m n$. 
We denote by ${\mathcal C}$ the set of codewords obtained by this convolutional code (in other words this is a convolutional code truncated in its first $N$ entries).
The codeword $\xb$ is then permuted by a fixed block interleaver $\pi$ of length $N$, we denote by $\yb$ the interleaved codeword. $\yb$ is then sent through a binary symmetric channel of crossover probability $p$. At the output of the channel we observe the noisy interleaved codeword $\zb$.

The blind identification process consists in observing $M$ noisy interleaved codewords $\zb^1, \dots, \zb^M$ to recover the convolutional code $\mathcal{C}$ and the block interleaver $\pi$. The codeword and interleaved codeword associated to $\zb^i$ are respectively denoted by $\xb^i$ and $\yb^i$. We also denote by $\mathcal{C}_\pi$ the code $\mathcal{C}$ interleaved by $\pi$ ($\yb^1, \dots, \yb^M$
belong to it).
We will assume that the length $N$ of the interleaver is known. It can obtained through the techniques given in  \cite{RLHY10,SH05a, SHB09, GBMN11} and recovering this length can now be considered to be a solved problem.
However, contrarily to \cite{LLG09b, GLLL13, JLLG12a, JLLG12b,JYLC11} we will make no assumption on the interleaver: it is chosen randomly among all permutations of size $N$. The convolutional code $\mathcal{C}$ is assumed to be unknown, its parameters $n$ and $k$ are also unknown.

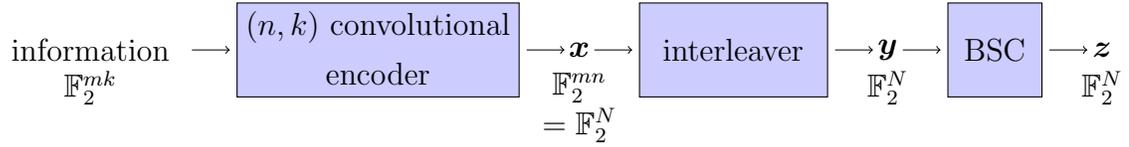
\begin{figure}
\begin{center}
\begin{tikzpicture}
        \tikzstyle discrete=[gray, very thin, font=\scriptsize]
        \def \cote{1.25}
        	\def \fleche{0.5}
        	\def \espace{0.1}

		\draw (0,0) node {information} ;
				\draw (0,-0.5) node {$\mathbb{F}_2^{mk}$} ;
		\draw [->, very thin] (\cote+\espace,0) -- (\cote+\fleche+\espace,0);
		\filldraw[fill=blue!20](\cote+\fleche+2*\espace, -0.5*\cote) rectangle  (4*\cote+\fleche+2*\espace, 0.5*\cote);
		
		\draw (2.5*\cote+\fleche+2*\espace,0.25*\cote) node {$(n,k)$ convolutional} ;
		\draw (2.5*\cote+\fleche+2*\espace,-0.25*\cote) node {encoder} ;		
		
		\draw [->, very thin] (4*\cote+\fleche+3*\espace,0) -- (4*\cote+2*\fleche+3*\espace,0);
		\draw (4*\cote+2*\fleche+5*\espace,0) node {$\xb$};
				\draw (4*\cote+2*\fleche+5*\espace,-0.5) node {$\mathbb{F}_2^{mn}$ };
				\draw (4*\cote+2*\fleche+5*\espace,-1) node{ $= \mathbb{F}_2^N$};
		\draw [->, very thin] (4*\cote+2*\fleche+7*\espace,0) -- (4*\cote+3*\fleche+7*\espace,0);
				
		\filldraw[fill=blue!20](4*\cote+3*\fleche+8*\espace, -0.5*\cote) rectangle node {interleaver} (6*\cote+3*\fleche+8*\espace, 0.5*\cote);

		\draw [->, very thin] (6*\cote+3*\fleche+9*\espace,0) -- (6*\cote+4*\fleche+9*\espace,0);

		\draw (6*\cote+4*\fleche+11*\espace,0) node {$\yb$};
				\draw (6*\cote+4*\fleche+11*\espace,-0.5) node {$\mathbb{F}_2^N$};

		\draw [->, very thin] (6*\cote+4*\fleche+13*\espace,0) -- (6*\cote+5*\fleche+13*\espace,0);

		\filldraw[fill=blue!20](6*\cote+5*\fleche+14*\espace, -0.5*\cote) rectangle node {BSC} (7*\cote+5*\fleche+14*\espace, 0.5*\cote);
		
		\draw [->, very thin] (7*\cote+5*\fleche+15*\espace,0) -- (7*\cote+6*\fleche+15*\espace,0);
		\draw (7*\cote+6*\fleche+17*\espace,0) node {$\zb$};
				\draw (7*\cote+6*\fleche+17*\espace,-0.5) node {$\mathbb{F}_2^N$};
		
\end{tikzpicture}
\caption{The communication scheme considered here.\label{fig:schema_codage}}
\end{center}
\end{figure}

\par{\em Our contribution. }
In this paper, we reconstruct the block interleaver $\pi$ and we recover the convolutional code $\mathcal{C}$. For this, we search for the dual code of $\mathcal{C}_\pi$. This code is recovered by known techniques \cite{CF09} which are adapted to recover parity-check equations of low weight given noisy codewords. 
Once $\mathcal{C}_\pi$ is recovered, we classify  the parity-check equations that have been obtained for $\mathcal{C}_\pi$ into groups.
Each group contains parity-check equations that correspond to parity-check equations of $\mathcal{C}$ whose positions differ by a multiple of $n$. We will 
then focus on a specific group of equations and will be able to reconstruct the convolutional code
and the unknown interleaver at the same time by introducing novel graph techniques in this setting.

By running some experimental tests we have been able to demonstrate the efficiency of this method. For example, for a convolutional code with parity-check equations of weight $6$, after finding a set of the parity-check equations, we reconstruct an interleaver of length $8000$ in less than ten seconds. The time for finding the interleaver once  $\mathcal{C}_\pi$ has been recovered does not depend on the noise level, the noise only impacts searching the  parity-check equations of $\mathcal{C}_\pi$. 
It should be added that the method used to recover these parity-check equations is more efficient than the methods calculating the rank of matrices, in particular when the data are noisy. This allows us to reconstruct efficiently the interleaver even for moderate noise levels.

\section{Overview of the algorithm }

Our reconstruction algorithm makes heavily use of two properties of the parity-check equations of an $(n,k)$ convolutional code
\begin{itemize}
\item[(i)]
for small up to moderate constraint length (which is the case of all convolutional codes used in practice) the parity-check 
equations are of low weight;
\item[(ii)] with the exception of a few parity-check equations involving the first bits, shifts of a parity-check equations by a multiple of $n$ are also parity-check equations of the convolutional code.
\end{itemize}

Interleaving does not destroy the first property but the second property is lost for $\CCpi$.
From now on we denote by $t$ the ``essential'' minimum weight of parity-check equations of $\CC$ (or $\CCpi$). By essential minimum weight 
we mean here that we take the minimum of the weights that appear at least a linear (in $N$) number of times.
We use this definition to discard parity-check equations that could involve the first bits of $\CC$ and that could be of lower weight due
to the zero initialization of the convolutional code. For the reconstruction we need a list $\mathcal{L}$ of these parity-check equations of weight $t$. 
 
Our algorithm basically works as follows
\begin{enumerate}
\item We use the algorithm of \cite{CF09} to find a list $\LC$ of parity-check equations of weight $t$.
\item We classify the parity-check equations of $\LC$ into disjoint groups $\LC_1,\dots,\LC_r$ such that two parity-check equations fall into the
same group if, and only if, they correspond to parity-check equations of $\CC$ which are shifts of each other by a multiple of $n$.
\item Denote by $\LC_1$ the group of parity-check equations of $\LC$ which have the smallest intersection number. The intersection number 
of a parity-check equation $\ec$ is the number of parity-check equations in $\LC$ which have at least one
position in common with $\ec$.
We use this group to recover one parity-check equation of $\CC$ by graph theoretic considerations.
\item We use this parity-check equation of $\CC$ to reorder the parity-check equations in $\LC_1$. Note that 
the structure of the group $\LC_1$ is such that these $\ell$ parity-check equations $\ec_1,\dots,\ec_\ell$
correspond to $\ell$ parity-check equations $\ec'_1,\dots,\ec'_\ell$ of $\CC$ which are shifts of a multiple of $n$ of each other.
This reordering is done in such a way that $\ec'_i$ is the shift by $n(i-1)$ of $\ec'_1$.
\item
This reordering of $\LC_1$ is then used in the last step 
to recover $\pi$.  
\end{enumerate}

\section{Notation}

A parity check equation $\ec$ will be denoted by the set of positions it involves, when we write
$\ec=\{e_1,\dots,e_t\}$ we mean here that this parity-check equation involves the positions $\{e_1,\dots,e_t\}$ of the
code that is considered (which is generally clear from the context).

With set notation applying an interleaver $\pi$ to code positions really amounts to transform a parity-check 
equation $\ec=\{e_1,\dots,e_t\}$ into a parity-check equation $\pi(\ec) \eqdef \{\pi(e_1),\dots,\pi(e_t)\}$.

\section{The reconstruction algorithm in detail}

\subsection{Recovering parity-check equations of weight $t$}
The first step consists in searching for a list $\mathcal{L}$ of parity-check equations of $\mathcal{C}_\pi$ of weight $t$. To obtain this list we apply the algorithm of \cite{CF09}. This method 
allows us to find the parity-check equations even if the observed codewords are noisy.
 
\subsection{Classifying parity-check equations into groups}

We want to classify parity-check equations of $\mathcal{L}$ into disjoint groups $\mathcal{L}_1, \dots, \mathcal{L}_r$ such that the parity-check equations in a group correspond to parity-check equations of $\mathcal{C}$ that are shifts of each other by a multiple of $n$. We say that these parity-check equations are of the same type. 

\begin{deff}[Type of a parity-check equation of $\mathcal{C}$]
$\ec=\{e_1,\dots, e_t\}$ and $\ec'=\{e'_1, \dots, e'_t\}$ two parity-check equations of $\mathcal{C}$ are of the same type if $\ec'$ is a shift by a multiple of $n$ of $\ec$. This means that there exists  $i$ such that $ \{e_1, \dots, e_t\}=\{e'_1+in, \dots, e'_t+in\}$. In such a case  we write $\ec \sim \ec'$.\\
All parity-check equations of the same type define an equivalence class. 
\end{deff}

\textit{Why classify ?} We need to classify parity-check equations of $\mathcal{L}$ because our method uses the regularity of parity-check equations of the convolutional code: shifts of parity-check equations by a multiple of $n$ are also parity-check equations of the convolutional code. A convolutional code can satisfy several types of parity-check equations of the same weight $t$.

\begin{example} The $(2,1)$ convolutional code which satisfies $x_1 + x_2 + x_3 + x_5 + x_6 + x_8= 0$ also satisfies its shifts: $\forall i$, $x_{1+2i} + x_{2+2i}+x_{3+2i} +x_{5+2i}+x_{6+2i}=0$. If we add two consecutive parity-check equations, we obtain another parity-check equation: $x_{1+2i}+x_{2+2i}+x_{4+2i}+x_{6+2i}+x_{7+2i}+
x_{10+2i} = 0$. This equation is verified for all integers $i$. So this code has at least two equivalence classes of parity-check equations: the first is represented by $\ec=\{1,2,3,5,6,8\}$ and the second by $\ec'=\{1,2,4,6,7,10\}$. The weight of all these parity-check equations is $6$. In this case, when we search for parity-check equations of weight $6$ of $\mathcal{C}_\pi$ we find equations corresponding to a mixture of these two types.
\end{example}

If we had directly parity-check equations of  $\mathcal{C}$ instead of parity-check equations of 
$\mathcal{C}_\pi$, then these different types of parity-check equation might get differentiated by their span: 

\begin{deff}[Span of a parity-chek equation]
Let $\ec=\{e_1, \dots, e_t\}$ be a parity-check equation, its span $s^\ec$ is defined by $s^\ec = \max_i (e_i) - \min_i (e_i)+1$.\\
In an equivalence class, all parity-check equations have the same span and we call this quantity the span of the equivalence class.
\end{deff}

Once interleaving this property is lost but the equivalence classes are always present:

\begin{deff}[Type of a parity-check equation of $\mathcal{C}_\pi$]
Two parity-check equations $\ec$ and $\ec'$ of $\mathcal{C}_\pi$ are of the same type if $\pi^{-1}(\ec) \sim \pi^{-1}(\ec')$.
\end{deff}

\textit{How to classify?} Even if we can not use the notion of the span of parity-check equations to 
classify the equations of $\mathcal{C}_\pi$, we will use the related notion of neighbourhood profile: 

\begin{deff}[Neighbourhood profile]
Let $\ec \in \mathcal{L}$, its neighbourhood profile $\mathcal{P}^\ec$ is a vector of length $t$: $\mathcal{P}^\ec = (\mathcal{P}^\ec_1, \dots, \mathcal{P}^\ec_t)$ where $\mathcal{P}^\ec_i=\#\{\ec' \in \mathcal{L}\text{ such as } |\ec \cap \ec'|=i\}$. 
\label{def:profiles}
\end{deff}

In other words, for a parity-check equation $\ec$, $\mathcal{P}^\ec_i$ is equal to the number of parity-check equations which have exactly $i$ positions in common with $\ec$. The number of parity-check equations with at least one position in common with $\ec$ defines its intersection number:

\begin{deff}[Intersection number]
The intersection number $\mathcal{I}^\ec$ of a parity-check equation $\ec \in \mathcal{L}$ is equal to $\mathcal{I}^\ec=\sum_{i\leq t}\mathcal{P}^\ec_i$. 
\end{deff}

\textit{Use profiles to determine the type of parity-check equations.} The point of Definition \ref{def:profiles} is that all parity-check equations of $\mathcal{C}$ of the same type have the same neighbourhood profile, whereas two equations of two different types have (in general) two different neighbourhood profiles. 
 It is also the case after interleaving the parity-check equations of $\mathcal{C}_\pi$.\\
Therefore we can classify parity-check equations into groups using their neighbourhood profiles. Of course, parity-check equations involving extreme positions of $\xb$ (the first or last) do not have exactly the same neighbourhood profile as the other parity-check equations of the same type. These parity-check equations have lost parity-check equations in their neighbourhood. This
motivates to bring in the following partial order on the profile of parity-check equations

\begin{deff}[Partial order on the profiles of parity-check equations]
We define a partial order: $\mathcal{P} \leq \mathcal{P}'$ if $\forall i \leq t$, $\mathcal{P}_i \leq \mathcal{P}'_i$ .
\end{deff}

\textit{Classifying a given parity-check equation.} 
The algorithm for classifying parity-check equations into groups is given by Algorithm \ref{algo:tri_equations}. With this algorithm we also deduce the length $n$ of the convolutional code $\mathcal{C}$.

\begin{algorithm}
  \caption{Classifying parity-check equations and deducing $n$ \label{algo:tri_equations}}
  {\bf input:} $\mathcal{L}$ a set of parity-check equations of $\mathcal{C}_\pi$\\
	{\bf  output:} \begin{itemize}
	\item $\mathcal{L}_1$ a set of parity-check equations of the same type
	\item
	the length $n$ of the convolutional code.
	\end{itemize}
  \For{all $\ec \in \mathcal{L}$}{
  	$\mathcal{P}^\ec \leftarrow $ the neigbourhood profiles of $\ec$
  }
  $\mathcal{P}^{E_1}, \dots, \mathcal{P}^{E_r} \leftarrow $ most frequent profiles in $\{\mathcal{P}^\ec, \ec \in \mathcal{L}\}$\\
  $\mathcal{L}_1, \dots, \mathcal{L}_r \leftarrow \emptyset$\\
  \For{all $\ec \in \mathcal{L}$}{
  	\If{$\mathcal{P}^{E_i}$ is the unique profile such that $\mathcal{P}^\ec \leq \mathcal{P}^{E_i}$}{
  		$\mathcal{L}_i \leftarrow \mathcal{L}_i \cup \{\ec\}$
  	}
  }
  \For{all $i \in \{1,\dots,r\}$}{
  	$\mathcal{I}^{E_i} \leftarrow \sum_{j\leq t} \mathcal{P}^{E_i}_j$
  }
  $\mathcal{L}_1 \leftarrow \mathcal{L}_i$ with $i$ such that $\mathcal{I}^{E_i} = \min_j \mathcal{I}^{E_j}$\\
  $n \leftarrow \lfloor \frac{N}{\#\mathcal{L}_1} \rfloor$\\
  \Return $\mathcal{L}_1$ and $n$
\end{algorithm}

\begin{remark}[Choose a group]
We form $r$ groups, but we just need one of them. We choose a group that minimizes the intersection number of its parity-check equations. This is a heuristic whose rationale is that the group with the smallest intersection number corresponds probably to the  equivalence class with the smallest span. Indeed, in $\mathcal{C}$, the larger the span of a parity-check equation $\ec$ is, the more chances we have that there  are parity-check equations with at least one position in common with $\ec$.\\
\end{remark}

\begin{remark}[Deducing the length $n$ of $\mathcal{C}$]
The number $nb_{eq}$ of parity-check equations in the group that we keep allows us to deduce the size $n$ of the convolutional code $\CC$, $n=\lfloor \frac{N}{nb_{eq}} \rfloor$. This equality is due  to the fact that almost all parity-check equations in this group correspond (after deinterleaving) to shifts by a multiple of $n$ of a single parity-check equation.
\end{remark}

\subsection{Recovering a parity-check equation of the convolutional code}

From now on, we assume that we have a set $\mathcal{L}_1$ of parity-check equations of $\mathcal{C}_\pi$. These parity-check equations are of weight $t$ and in the same equivalence class: they correspond to parity-checks of $\mathcal{C}$ which are shifts of each other by a multiple of $n$.\\

We denote by $\ec_{\mathcal{C}}$ a parity-check equation of $\mathcal{C}$ such that each parity-check equation $\ec$ of $\mathcal{L}_1$ satisfies $\pi^{-1}(\ec) \sim \ec_{\mathcal{C}}$ (that is each parity-check equation $\ec$ of 
$\mathcal{L}_1$ is such that  $\pi^{-1}(\ec)$ is a shift of $\ec_{\mathcal{C}}$). 

The purpose of this subsection is to show how $\ec_{\mathcal{C}}$ can be recovered from the knowledge of $\mathcal{L}_1$.
$\ec_\mathcal{C}$ is the parity-check equation of a sub-code of $\mathcal{C}$, this sub-code is an $(n, n-1)$ convolutional code. To recover this $(n, n-1)$ convolutional code we test each $(n, n-1)$ convolutional code 
that admits a parity-check equation of weight $t$ and with a span less than $\smax$ where $\smax$ is some
chosen constant.
Our strategy is to attach a graph to a set of parity-check equations such that \\
(i) the equivalence class of a parity-check equation $\ec$ of an $(n,n-1)$ convolutional code discriminates the convolutional code\\
(ii) if two sets of parity-check equations differ from each other by a permutation then their associated graphs are isomorphic.

By checking if there is an isomorphism between the graph associated to $\mathcal{L}_1$ and 
the graph associated to shifts of a parity-check equation of an $(n,n-1)$ convolutional code
we will recover the right convolutional code and adding labels to the graph will allow us to
identify the permutation between the two sets of parity-check equations.

\subsubsection*{Graphs associated to $\LC_1$ and $\ec$}
From now on we will use the following notation

\textbf{Notation}
We denote by $\ell$ the number of parity-check equations in $\LC_1$.\\

To the set of parity-check equations $\LC_1$ we associate a labeled graph $\Gt{\LC_1}$ which is defined as follows

\begin{deff}[Graph associated to a set of parity-check equations]
\label{def:graph}
The labeled graph $\Gt{\LC}$ associated to a set $\LC$ of parity-check equations is such that
\begin{itemize}
\item Each parity-check equation of $\LC$ is represented in $\Gt{\LC}$ by a vertex.
\item If $\ec$ and $\ec'$, two parity-check equations of $\LC$, have $k$ positions in common (that is $|\ec \cap \ec'|=k$) then in $\Gt{\LC}$ the two corresponding vertices are connected by $k$ edges.
\item Each edge of $\Gt{\LC}$ is labeled with the number of the position that it represents.
\end{itemize}
When $\LC$ is clear from the context we will just denote this graph by $\Gtilde$.
\end{deff}

\textbf{Notation} We denote by $\G$ the graph $\Gtilde$ without label on edges.\\

\begin{example}
Let $\LC = \{\ec_1, \ec_2, \ec_3, \ec_4\}$ with $\ec_1=\{1,4,6\}$, $\ec_2=\{2,4,5\}$, $\ec_3=\{4,6,7\}$ and $\ec_4=\{2,5,7\}$. The graph $\Gtilde$ associated to $\LC$ is represented on Figure \ref{ex:ensembleEqGraph}.
\end{example}

\begin{figure}
\centering
\begin{tikzpicture}[
  >=stealth',
  shorten >=1pt,
  auto,
  node distance=1.8cm,
  thick,
  main node/.style={circle, thin,fill=blue!20,draw,font=\sffamily\bfseries}
  ]

\node[main node](1){$\ec_1$};
\node[main node](2)[above right of=1]{$\ec_2$};
\node[main node](3)[below right of=1]{$\ec_3$};
\node[main node](4)[above right of=3]{$\ec_4$};


\path[every node/.style={font=\sffamily\footnotesize}, thin]
(1) edge [bend left=15] node[right]{$4$} (3)
	edge [bend right=15] node[left]{$6$} (3)
	edge [bend left=10] node{$4$} (2)

(2) edge [bend left=15] node{$4$} (3)
	edge [bend left=15] node[right]{$2$} (4)
	edge [bend right=15] node[right]{$5$} (4)

(3) edge [bend right=15] node[below]{$7$}(4);
\end{tikzpicture}
\caption{Graph $\Gtilde$ associated to $\LC=\{\ec_1=\{1,4,6\},\ec_2=\{2,4,5\},\ec_3=\{4,6,7\},\ec_4=\{2,5,7\}\}$ \label{ex:ensembleEqGraph}}
\end{figure}

The graph $\Gt{\LC_1}$ associated to $\LC_1$ represents the interleaved sub-code of $\CC$. To recover this sub-code (not interleaved) we test each $(n, n-1)$ convolutional code. This is achieved as follows.
An $(n, n-1)$ convolutional code is defined by a parity-check equation $\ec=\{e_1, \dots, e_t\}$. Using $\ec$ we construct a set $\LC_\ec$ of parity-check equations of this $(n,n-1)$ convolutional code: $\LC_\ec=\{\{e_1+in, e_2+in, \dots, e_t+in\}, -\frac{\ell}{n} \leq i < \frac{\ell}{n}\}$. $\LC_\ec$ contains $\ell$ consecutive parity-check equations obtained by shifts of $\ec$ by a multiple of $n$. Using Definition \ref{def:graph} we associate the graph $\Gt{\LC_\ec}$ to this set $\LC_\ec$. To simplify notation we denote this graph by $\Gtilde^\ec$.\\
We want to compare this graph $\Gtilde^\ec$ to $\Gt{\LC_1}$, it is for this reason that we take $\ell$ parity-check equations in $\LC_\ec$, so the two graphs have the same number of vertices, and we check if they are isomorphic. 

\begin{deff}[Isomorphic graphs]
Two graphs $\Gtilde$ and $\Gtilde'$ are isomorphic if, and only if, there exists a bijective mapping $\phi$ between the vertices of $\Gtilde$ and the vertices of $\Gtilde'$ so that for any pair of vertices $(x,y)$ of $\Gtilde$ there is the 
same number of edges between $x$ and $y$ as there are edges between $\phi(x)$ and $\phi(y)$ in 
$\Gtilde'$.\end{deff}

We we also need a finer definition of isomorphism which is suitable for labeled graphs
\begin{deff}[Equivalent graphs]
Two labeled graphs $\Gtilde$ and $\Gtilde'$ are equivalent if they are isomorphic (call the corresponding mapping $\phi$) and there exists a bijective mapping $\psi$ 
of the labels from one graph to the other so that for any pair of vertices $(x,y)$ of $\Gtilde$ if we denote 
by $\{a_1,\dots,a_s\}$ the (multi)set of labels of the edges between $x$ and $y$, then the edges
between $\phi(x)$ and $\phi(y)$  in $\Gtilde'$ have labels $\{\psi(a_1),\dots,\psi(a_s)\}$.\\
\end{deff}

To recover the parity-check equation $\ec_\CC$ of the sub-code of $\CC$ and the interleaver $\pi$ we use the following proposition

\begin{proposition}
If $\ec = \ec_\CC$ then $\Gtilde^\ec$ and $\Gt{\LC_1}$ are equivalent.
\end{proposition}

\begin{proof}
We assume that $\ec = \ec_\CC$. Let $\LC_\ec = \{\ec_0, \dots, \ec_{\ell-1}\}$ be the set of parity-check equations associated to $\ec$, with $\ec_i$ being equal to $\ec_{i-1}$ shifted by $n$. $\ec=\ec_\CC$ so $\LC_1$ contains the same parity-check equations as $\LC_\ec$ but interleaved by the interleaver $\pi$: $\LC_1=\{\pi(\ec_0), \pi(\ec_1), \dots, \pi(\ec_{\ell-1})\}$. Note that the interleaver changes the numbering of positions, not the number of positions in common between two parity-check equations. An isomorphism $\phi$ between vertices of this graph is given by $\phi : \Gtilde^\ec \rightarrow \Gt{\LC_1}$, $\ec_i \mapsto \pi(\ec_i)$, $\forall i < \ell-1$. This shows that that $\Gtilde^\ec$ and $\Gt{\LC_1}$ are isomorphic. If we denote by $m^\ec$ the minimal value such that $\Gtilde^\ec$ contains a vertex representing a parity-check equation involving the position $m^\ec$, the block interleaver $\pi$ gives an isomorphism $\psi$ on labels between $\Gtilde^\ec$ and $\Gtilde^\pi$. $\psi : \Gtilde^\ec \rightarrow \Gtilde(\LC_1)$, $i \mapsto \pi(i-m^\ec)$. We obtain that $\Gtilde^\ec$ and $\Gtilde(\LC_1)$ are equivalent.
\end{proof}

\begin{remark}
If we find a parity-check equation $\ec$ such that $\Gt{\LC_1}$ and $\Gtilde^\ec$ are equivalent, then an isomorphism $\psi$ between labels of these graphs gives the block interleaver $\Pi$ such that $\Pi(\CC)=\CC_\pi$.
\end{remark}

\subsubsection*{Sub-graphs associated to $\LC_1$ and $\ec$}
To check the equivalence we will need auxiliary graphs which are much smaller and that will
in general be sufficient for testing the equivalence between graphs.
More precisely, we use sub-graphs induced by $\Gtilde^\ec$ and $\Gt{\LC_1}$.\\

{\bf Notation} From now on to simplify notation we will denote the graph 
$\Gt{\LC_1}$ by $\Gtilde^\pi$.\\

We will associate six  sub-graphs, $\G^\pi_1, \G^\pi_2, \Gtilde^\pi_2, \G^\ec_1, \G^\ec_2$ and $\Gtilde^\ec_2$, to $\LC_1$ and $\ec$  such that if $\Gtilde^\ec$ and $\Gtilde^\pi$ are equivalent then:
\begin{itemize}
\item $\G^\ec_1$ and $\G^\pi_1$ are isomorphic
\item $\G^\ec_2$ and $\G^\pi_2$ are isomorphic
\item $\Gtilde^\ec_2$ and $\Gtilde^\pi_2$ are equivalent
\end{itemize}

The first graphs $\G^\ec_1$ and $\G^\pi_1$ are not labeled and represent the neighbourhood  of a parity-check equation.\\

To obtain $\G^\pi_1$, we randomly choose a parity-check equation $\ec_0$ in $\LC_1$. $\G^\pi_1$ is the sub-graph of $\G^\pi$ induced by the vertex representing $\ec_0$ and all vertices having at least one edge in common with it.\\
$\G^\ec_1$ is a sub-graph of $\G^\ec$ induced by vertices representing $\ec$ and all its shifts by a multiple of $n$ such that they have at least one position in common with $\ec$. This graph contains only a small number of vertices as shown by

\begin{proposition}
Let $\ec$ be a parity-check equation of an $(n, n-1)$ convolutional code and $s^\ec$ the span of $\ec$. The sub-graph $\G^\ec_1$ associated to $\ec$ contains at most $2\lceil \frac{s^\ec}{n}\rceil -1$ vertices.
(In other words, the parity-check equation $\ec$ has at most $2\lceil \frac{s^\ec}{n}\rceil -1$ parity-check equations in its neighbourhood.)
\end{proposition}

\textbf{Notation} We denote by $\ec^{(i)}$ the parity-check equation equals to $\ec$ shifted by $in$.

\begin{proof}
Assume that the parity-check equation $\ec$ is given by $\ec=\{e_1, \dots, e_t\}$ with $e_1 < e_2 < \dots < e_t$. $s^\ec$ is the span of $\ec$, so $s^\ec = e_t-e_1+1$. $\ec^{(i)}$ is represented in $\G_1^\ec$ if and only if $\ec^{(i)}$ and $\ec$ have at least one position in common, that is $\{e_1, \dots, e_t\} \cap \{e_1+in, \dots, e_t+in\} \neq \emptyset$.\\
For $i \geq 0$ we have $\{e_1, \dots, e_t\} \cap \{e_1+in, \dots, e_t+in\} \neq \emptyset$ if $e_t \geq e_1+in$, that is $0\leq i < \frac{s^\ec}{n}$. So for $i \geq 0$, there are at most $\lceil \frac{s^\ec}{n}\rceil$ parity-check equations which can have positions in commons with $\ec$.\\
If $i < 0$. $\{e_1, \dots, e_t\} \cap \{e_1+in, \dots, e_t+in\} \neq \emptyset$ if $e_t+in \geq e_1$, that is $-\frac{s^\ec}{n} < i < 0$. For $i < 0$, there is at most $\lceil \frac{s^\ec}{n} \rceil-1$ parity-check equations which can have positions in commons with $\ec$.\\
So, the maximal number of parity-check equations which can have a position in common with $\ec$ is equal to $2\lceil \frac{s^\ec}{n}\rceil -1$.
\end{proof}

\begin{proposition} If $\pi^{-1}(\ec_0)$ does not involve the first or last positions, and if $\ec=\ec_\CC$ then $\G_1^\ec$ and $\G_1^\pi$ are isomorphic.
\end{proposition}

\begin{proof}
$\ec=\ec_\CC$ (that is $\pi^{-1}(\ec_0) \sim \ec$), we denote by $\ec'$ the parity-check equation $\ec^{(k)}$ such that $\pi^{-1}(\ec_0)=\ec^{(k)}$. If $\pi^{-1}(\ec_0)$ does not involve the first or last positions and if we denote by $I=\{i_1, \dots, i_j\}$ the set of integers such as $\forall i \in I$, $\ec'$ and $\ec'^{(i)}$ have at least one position in common, then $\LC_1$ contains $\pi(\ec'^{(i_1)}), \dots, \pi(\ec'^{(i_j)})$. All these parity-check equations have at least one position in common with $\ec_0$, so they are represented in $\G_1^\pi$ and we have an isomorphism $\phi$ between $\G_1^\ec$ and $\G_1^\pi$ defined by $\phi : \G_1^\ec \rightarrow \G_1^\pi$, $\ec^{(i)} \mapsto \pi(\ec^{(i)})$.
\end{proof}

So the first step to test a given $(n, n-1)$ convolutional code, consists in checking if $\G_1^\ec$ and $\G_1^\pi$ are isomorphic. If $\G_1^\ec$ and $\G_1^\pi$ are not isomorphic then $\ec \neq \ec_\CC$. But these graphs are not enough discriminating, two $(n, n-1)$ convolutional codes, defined by $\ec$ and $\ec'$ can be associated to two isomorphic graphs $\G_1^\ec$ and $\G_1^{\ec'}$.

\begin{example}
For $n=2$, the graph $\G_1^\ec$ associated to the parity-check equation $\ec=\{1,2,4,6,7\}$ is isomorphic to the graph $\G_1^{\ec'}$ associated to the parity-check equation $\ec'=\{1,3,4,6,7\}$. These graphs are represented on Figure \ref{fig:ex:graphesIsomorphes}. An isomorphism between these graphs is defined by $\phi : \G_1^\ec \rightarrow \G_1^{\ec'}$, $\ec^{(i)} \mapsto \ec'^{(i)}$.\\
\end{example}

\begin{figure}
\centering
\begin{tikzpicture}[
  >=stealth',
  shorten >=0.5pt,
  auto,
  node distance=2.0cm,
  thick,
  main node/.style={circle, thin,fill=blue!20,draw,font=\sffamily\bfseries}
  ]

\node[main node](1){$\ec^{(-3)}$};
\node[main node](2)[below right of=1]{$\ec^{(-2)}$};
\node[main node](3)[above right of=2]{$\ec^{(-1)}$};
\node[main node](4)[below  right of=3]{$\ec^{(0)}$};
\node[main node](5)[above right of=4]{$\ec^{(1)}$};
\node[main node](6)[below  right of=5]{$\ec^{(2)}$};
\node[main node](7)[above right of=6]{$\ec^{(3)}$};

\node[main node](8)[below left of=2]{$\ec'^{(-3}$};
\node[main node](9)[below right of=8]{$\ec'^{(-2)}$};
\node[main node](10)[above right of=9]{$\ec'^{(-1)}$};
\node[main node](11)[below  right of=10]{$\ec'^{(0)}$};
\node[main node](12)[above right of=11]{$\ec'^{(1)}$};
\node[main node](13)[below  right of=12]{$\ec'^{(2)}$};
\node[main node](14)[above right of=13]{$\ec'^{(3)}$};

\path[every node/.style={font=\sffamily\footnotesize}, thin]
(1) edge [bend left=15] (2)
	edge [bend right=15] (2)
	edge [bend left=10] (3)
	edge (4)
	
(2) edge [bend left=15] (3)
	edge [bend right=15] (3)
	edge [bend right=10] (4)
	edge (5)

(3) edge [bend left=15] (4)
	edge [bend right=15] (4)
	edge [bend left=10] (5)
	edge (6)
	
(4) edge [bend left=15] (5)
	edge [bend right=15] (5)
	edge [bend right=10] (6)
	edge (7)
	
(5) edge [bend left=15] (6)
	edge [bend right=15] (6)
	edge [bend left=10] (7)
	
(6) edge [bend left=15] (7)
	edge [bend right=15] (7)
	
(8) edge [bend left=15] (9)
	edge [bend right=15] (9)
	edge [bend left=10] (10)
	edge (11)
	
(9) edge [bend left=15] (10)
	edge [bend right=15] (10)
	edge [bend right=10] (11)
	edge (12)

(10) edge [bend left=15] (11)
	edge [bend right=15] (11)
	edge [bend left=10] (12)
	edge (13)
	
(11) edge [bend left=15] (12)
	edge [bend right=15] (12)
	edge [bend right=10] (13)
	edge (14)
	
(12) edge [bend left=15] (13)
	edge [bend right=15] (13)
	edge [bend left=10] (14)
	
(13) edge [bend left=15] (14)
	edge [bend right=15] (14);
\end{tikzpicture}
\caption{Graphs $\G_1^{\ec}$ and $\G_1^{\ec'}$ with $n=2$, $\ec=\{1,2,4,6,7\}$ and $\ec'=\{1,3,4,6,7\}$. \label{fig:ex:graphesIsomorphes}}
\end{figure}

We associate to $\LC_1$ and $\ec$ two other graphs $\G_2^\pi$ and $\G_2^\ec$. These graphs are not labeled and represent the neighbourhood at distance two of a parity-check equation. 

$\G_2^\pi$ is the sub-graph of $\G^\pi$ induced by $\G_1^\pi$ and all vertices having at least one edge in common with a vertex of $\G_1^\pi$. So $\G_2^\pi$ represents the neighbourhood at distance 2 of $\ec_0$ in $\LC_1$.\\
$\G_2^\ec$ is the sub-graph of $\G^\ec$ induced by $\G_1^\ec$ and all vertices having at least one edge in common with a vertex of $\G_1^\pi$. This graph is rather small too as shown by:\\

\begin{proposition}
Let $\ec$ be a parity-check equation of an $(n, n-1)$ convolutional code and $s^\ec$ be the span of $\ec$. The sub-graph $\G^\ec_2$ associated to $\ec$ contains at most $4\lceil \frac{s^\ec}{n}\rceil -3$ vertices. (In other words, the parity-check equation $\ec$ has at most $4\lceil \frac{s^\ec}{n}\rceil -3$ parity-check equations in its neighbourhood at distance 2.)
\end{proposition}

\begin{proposition}
\label{prop:graphesDistance2Iso}
If $\pi^{-1}(\ec_0)$ does not involve the first or last positions, and if $\ec=\ec_\CC$ then $\G_2^\ec$ and $\G_2^\pi$ are isomorphic.
\end{proposition}

\begin{proof}
$\ec=\ec_\CC$ (that is $\pi^{-1}(\ec_0) \sim \ec$), we denote by $\ec'$ the parity-check equation $\ec^{(k)}$ such that $\pi^{-1}(\ec_0)=\ec^{(k)}$. If $\pi^{-1}(\ec_0)$ does not involve the first or last positions and if we denote by $I=\{i_1, \dots, i_j\}$ the set of integers such that for all $ i$ in $I$, $\Gtilde_2^\ec$ contains a vertex representing $\ec^{(i)}$, then $\LC_1$ contains $\pi(\ec'^{(i_1)}), \dots, \pi(\ec'^{(i_j)})$. All these parity-check equations are in the neighbourhood at distance 2 of $\ec_0$, so they are represented in $\G_2^\pi$ and we have an isomorphism $\phi$ between $\G_2^\ec$ and $\G_2^\pi$ defined by $\phi : \G_2^\ec \rightarrow \G_2^\pi$, $\ec^{(i)} \mapsto \pi(\ec^{(i)})$.
\end{proof}

The second step to test a given $(n, n-1)$ convolutional code, consists in checking if $\G_2^\ec$ and $\G_2^\pi$ are isomorphic but these graphs are not sufficiently discriminating. Finally we use two small labeled graphs $\Gtilde_2^\ec$ and $\Gtilde_2^\pi$.\\

To obtain $\Gtilde_2^\ec$ and $\Gtilde_2^\pi$ we just add label on edges of $\G_2^\ec$ and $\G_2^\pi$.\\

\begin{proposition}
If $\pi^{-1}(\ec_0)$ does not involve the first or last positions, and if $\ec=\ec_\CC$ then $\Gtilde_2^\ec$ and $\Gtilde_2^\pi$ are equivalent.
\end{proposition}

\begin{proof}
With Proposition \ref{prop:graphesDistance2Iso} we deduce that $\Gtilde_2^\ec$ and $\Gtilde_2^\pi$ are equivalent. If we denote by $m^\ec$ and $m^\pi$ the minimal values such that $\Gtilde_2^\ec$ and $\Gtilde_2^\pi$ contain a vertex representing a parity-check equation involving respectively positions $m^\ec$ and $m^\pi$, then we have an isomorphism $\psi$ between labels of $\Gtilde_2^\ec$ and $\Gtilde_2^\pi$ : $\psi : \Gtilde_2^\ec \rightarrow \Gtilde_2^\pi$, $i \mapsto \pi(i-m^\ec+\pi^{-1}(m^\pi))$.
\end{proof}

Finally the algorithm used for recovering the parity-check equation $\ec_\CC$ of the sub-code of $\CC$ is Algorithm \ref{algo:reconnaissance_eq}.

\begin{algorithm}
  \caption{Recovering the parity-check equation $\ec_\CC$ \label{algo:reconnaissance_eq}}
  {\bf input:} $\mathcal{L}_1$ a set of parity-check equations of weight $t$ of the same type 
and $n$ the length of $\mathcal{C}$ \\
	{\bf  output:} $L$ a list of parity-check equations such that $\ec_\CC$ can be equal to each of them\\
	$L \leftarrow \emptyset$\\
	$\ec_0 \leftarrow $ choose at random a parity-check equation of $\mathcal{L}_1$\\
	$\G_1^\pi, \G_2^\pi$ and $\Gtilde_2^\pi \leftarrow$ sub-graphs induced by $\ec_0$ and its neighbourhood\\	
  \For{all $\ec$ of weight $t$ and with a span less than $s_{max}$}{
 	$\G_1^\ec \leftarrow$ graph representing the neighbourhood of $\ec$ at distance 1\\
 	\If{$\mathcal{G}_1^\ec$ and $\mathcal{G}_1^\pi$ are isomorphic}{
 		$\mathcal{G}_2^\ec \leftarrow$ graph representing the neighbourhood at distance 2 of $\ec$\\
 		\If{$\mathcal{G}_2^\ec$ and $\mathcal{G}_2^\pi$ are isomorphic}{
 			$\Gtilde_2^\ec \leftarrow$ labeled graph representing the neighbourhood at distance 2 of $\ec$\\
 			\If{$\Gtilde_2^\ec$ and $\Gtilde_3^\pi$ are equivalent}{
 				$L \leftarrow L \cup \{\ec\}$\\
 				}
 			}
 		}
 	}
 	
  \Return $L$
\end{algorithm}

\subsubsection*{Reducing the number of tests}

In fact, these graphs have lots of symmetries, and we do not really need to test all parity-check equations of weight $t$ and with a span less than $s_{max}$.\\
The following proposition allows us to reduce the number of $(n, n-1)$ convolutional code that we have to test.\\

\begin{proposition}
Let $\ec = \{e_1, \dots, e_t\}$ be the parity-check equation of an $(n, n-1)$ convolutional code and $\Gtilde_2^\ec$ be the labelled graph representing the neighbourhood at distance two of $\ec$.\\
$\ec$ can also be represented by a binary vector $b_1\dots b_s$ where $b_i=1$ if $i \in \{e_1, \dots, e_t\}$.

\begin{itemize}
\item The graph $\Gtilde_2^{\ec'}$ associated to $\ec'$ represented by the binary vector $b_s\dots b_1$ is equivalent to $\Gtilde_2^{\ec}$.
\item For all permutations $p=[p_1, \dots, p_n]$ of length $n$, the graph $\Gtilde_2^{\ec'}$ associated to the parity-check equation $\ec'$ represented by the binary vector\\ $p(b_1\dots b_n)p(b_{n+1}\dots b_{2n})\dots p(b_{s-n+1}\dots b_s)$ is equivalent to $\Gtilde_2^\ec$.
\end{itemize}
\label{prop:eqEquivalentes}
\end{proposition}

\begin{proof}
\begin{itemize}
\item For the first point, if $\Gtilde_2^{\ec}$ contains $\ec^{(i_1)}, \dots, \ec^{(i_j)}$ then $\Gtilde_2^{\ec'}$ contains $\ec'^{(i_1)}, \dots, \ec'^{(i_j)}$, and between these two graphs we have the isomorphism $\phi$ defined by $\phi : \Gtilde_2^{\ec} \rightarrow \Gtilde_2^{\ec'}$, $\ec^{(i)} \mapsto \ec'^{(i_j-i)}$ for all $i \in \{i_1, \dots, i_j\}$. The isomorphism $\psi$ between labels of these graphs can be defined by $\psi : \Gtilde_2^{\ec} \rightarrow \Gtilde_2^{\ec'}$, $k \mapsto m^\ec-k$ where $m^\ec$ is the maximal value of labels of $\Gtilde_2^{\ec'}$. With these two isomorphisms we deduce that $\Gtilde_2^{\ec}$ and $\Gtilde_2^{\ec'}$ are equivalent.
\item For the second point, if $\Gtilde_2^{\ec}$ contains $\ec^{(i_1)}, \dots, \ec^{(i_j)}$ then $\Gtilde_2^{\ec'}$ contains $\ec'^{(i_1)}, \dots, \ec'^{(i_j)}$, and between these two graphs we have the isomorphism $\phi$ defined by $\phi : \Gtilde_2^{\ec} \rightarrow \Gtilde_2^{\ec'}$, $\ec^{(i)} \mapsto \ec'^{(i)}$ for all $i \in \{i_1, \dots, i_j\}$. We define the permutation $P$ by $P(i) = p(i \mod n)+ \lfloor\frac{i}{n}\rfloor$. The isomorphism $\psi$ on labels defined by $\psi : \Gtilde_2^{\ec} \rightarrow \Gtilde_2^{\ec'}$, $k \mapsto P(k)$ allows us to deduce that $\Gtilde_2^{\ec}$ and $\Gtilde_2^{\ec'}$ are equivalent.
\end{itemize}
\end{proof}

\begin{deff}[Equivalent parity-check equations]
Let $\ec$ and $\ec'$ be two parity-check equations, if using the Proposition \ref{prop:eqEquivalentes} we can deduce that the two graphs $\Gtilde_2^\ec$ and $\Gtilde_2^{\ec'}$ are equivalent we say that $\ec$ and $\ec'$ are equivalent.
\end{deff}

\begin{example}
For $n = 2$, if $s_{max}=20$ and $t=10$ we run only $15\,328$ tests instead of $184\,756$. If we suppose that $s_{max}=30$, $1\,238\,380$ tests are needed instead of $30\,045\,015$.\\
If the sought parity-check equation is $\ec_\CC=\{1,2,3,5,6,7,8,12,13,14\}$ (of weight 10), only 2 parity-check equations produce an isomorphic graph to $\mathcal{G}_1^\pi$ and among them one is equivalent to $\Gtilde_2^\pi$ for $s_{max}=20$ (testing the $15\,328$ parity-check equations takes approximately 1 second). If we take $s_{max}=30$, 4 graphs are isomorphic to $\mathcal{G}_1^\pi$ and 2 are equivalent to $\Gtilde_2^\pi$, (one of them is eliminated later) these tests take less than 3 minutes.\\
\end{example}

\subsubsection*{If the parity-check equation $\ec_\mathcal{C}$ is not recovered}
If no parity-check equation has an equivalent labeled graph with $\Gtilde_2^\pi$ we may have chosen in $\mathcal{L}_1$ a parity-check equation $\ec_0$ which has incomplete graphs (after deinterleaving this parity-check equation involves the first or last positions of $\xb$ or  at least a parity-check equation in its neighbourhood at distance $2$
is missing in $\mathcal{L}_1$).\\

In this case, we randomly choose another parity-check equation in $\mathcal{L}_1$, we compute the new graphs $\mathcal{G}_1^\pi$, $\mathcal{G}_2^\pi$ and $\Gtilde_2^\pi$ representing its neighbourhood at distance one and two, and we test all convolutional codes. \\

\begin{remark}
If $\mathcal{G}_1^\pi$ or $\mathcal{G}_2^\pi$ is incomplete, it is probably not symmetric so, no graph $\mathcal{G}_1^\ec$ or $\mathcal{G}_2^\ec$ can be isomorphic with it and the test of all $(n, n-1)$ convolutional codes is very quickly (we just compare the number of vertices, they have not the same number so they can't be isomorphic).\\
\end{remark}

If $\ec_\mathcal{C}$ can be equal to several parity-check equations $\ec$ we apply the end of the method for each of them.

\subsection{Ordering parity-check equations}
Using $\ec_{\mathcal{C}}$ the parity-check equation previously recovered, we want to order the parity-check equations of $\mathcal{L}_1$. That is, find an ordering $\mathcal{A}=\ec_{a_1}, \dots, \ec_{a_l}$ of these parity-check equations such that $\pi^{-1}(\ec_{a_{i+1}})$ is equal to the shift by $n$ of $\pi^{-1}(\ec_{a_i})$. All parity-check equations of $\mathcal{L}_1$ belong to $\mathcal{A}$ once and only once.\\

To order these parity-check equations we extend the two graphs $\mathcal{G}_2^\pi$ and $\mathcal{G}_2^{\ec_\mathcal{C}}$ and we search for an isomorphism between the vertices of these two extended graphs. This isomorphism give us the ordering $\mathcal{A}$. \\

When we recover the parity-check equation $\ec_\mathcal{C}$ of the sub-code of $\mathcal{C}$ we search for an isomorphism between $\mathcal{G}_2^\pi$ and $\mathcal{G}_2^{\ec_\mathcal{C}}$. Once we know this isomorphism, we also have the bijection $\phi$ between the vertices of these  graphs. This bijection gives us a part of the ordering. Indeed, for all $i$, such that $\mathcal{G}_2^{\ec_{\mathcal{C}}}$ contains a vertex $V_i$ representing $\ec_\mathcal{C}$ shifted by $in$, we place the parity-check equation represented by $\phi(V_i)$ at position $i$ in $\mathcal{A}$. \\

To obtain the bijection using all parity-check equations of $\mathcal{L}_1$ and deduce the entire ordering $\mathcal{A}$, we extend step by step $\mathcal{G}_2^\pi$, $\mathcal{G}_2^{\ec_\mathcal{C}}$ and the bijection $\phi$. We denote by $\G_{a..b}^{\ec_\CC}$ the graph representing $\ec$ shifted by $in$ for all integers $i \in [a,b]$, $\phi_{a..b}$ the isomorphism defined for all integers between $a$ and $b$, and $\G_{a..b}^{\pi}$ the graph $\phi(\G_{a..b}^{\ec_\CC})$.\\

\textit{A step of the extension.}
Knowing $\G_{a..b}^{\ec_\CC}$, $\G_{a..b}^{\pi}$ and $\phi_{a..b}$, we search for $\G_{a..b+1}^{\ec_\CC}$, $\G_{a..b+1}^{\pi}$ and $\phi_{a..b+1}$.

\begin{itemize}
\item To obtain $\G_{a..b+1}^{\ec_\CC}$ from $\G_{a..b}^{\ec_\CC}$ we just add a vertex representing $\ec_\CC$ shifted by $(b+1)n$ and the corresponding edges.
\item We search  in $\LC_1$ for a parity-check equation $\ec_i$ which is not represented in $\G_{a..b}^{\pi}$ and such that if we add a vertex representing this parity-check equation and the corresponding edges to $\G_{a..b}^{\pi}$, then $\phi_{a..b+1}$ defined by $\phi_{a..b+1}(j)=\psi_{a..b}(j)$ for $j \in[a,\dots, b]$ and $\phi_{a..b+1}(b+1)=i$ is an isomorphism between $\G_{a..b+1}^{\ec_\CC}$ and $\G_{a..b}^{\pi}$ extended with $\ec_i$.
\end{itemize}

When we can not extend $\G_{a..b}^{\ec_\CC}$, $\G_{a..b}^{\pi}$ and $\phi_{a..b}$ such that $\G_{a..b+1}^{\ec_\CC}$ and $\G_{a..b+1}^{\pi}$ are isomorphic, we extend these graphs and the isomorphism in the other direction. In other words, we search for $\G_{a-1..b}^{\ec_\CC}$, $\G_{a-1..b}^{\pi}$ and $\phi_{a-1..b}$ from $\G_{a..b}^{\ec_\CC}$, $\G_{a..b}^{\pi}$ and $\phi_{a..b}$.\\

\begin{remark}[Several parity-check equations]
If at a given step, several parity-check equations $\ec_i$ of $\LC_1$ can be chosen, then we extend the two graphs and the isomorphism with feedback and finally we choose the biggest isomorphism and corresponding graphs.\\
\end{remark}

\begin{remark}[No parity-check equation]
If no parity-check equation in $\mathcal{L}_1$ satisfies all conditions, it might be that the sought parity-check equation is not in $\mathcal{L}_1$. This parity-check equation was not found using \cite{CF09}, or not classified in this group (it is an unclassified parity-check equation). In this case, we add a "missing parity-check equation" to $\mathcal{G}_{a..b}^{\pi}$ , that is we add a vertex and edges to respect the regularity of $\mathcal{G}_{a..b}^{\pi}$ and we define $\phi_{a..b+1}(b+1) = "missing"$ or $\phi_{a-1..b}(a-1) = "missing"$. Then we continue the extension of graphs and $\phi$.
\end{remark}

\begin{remark} We do not add more than $\lceil \frac{s^\ec}{n} \rceil -1$ consecutive "missing" parity-check equations in $\mathcal{A}$ because in this case, the next parity-check equation has no position in common with previous parity-check equations.
\end{remark}

\begin{example}
We represent on Figures \ref{fig:exemple_eqManqute1}, \ref{fig:exemple_eqManqute2} and \ref{fig:exemple_eqManqute3} the extension with a missing parity-check equation (in red). There is an edge connecting a vertex lying before the missing parity-check equation to a vertex lying after this missing parity-check equation.
\end{example}

\begin{figure}
\begin{center}
\begin{tikzpicture}[
  >=stealth',
  shorten >=0.8pt,
  auto,
  node distance=0.8cm,
  thick,
  main node/.style={circle, thin,fill=blue!20,draw,font=\sffamily\bfseries}
  ]

\node[main node](2){};
\node[main node](3)[above right of=2]{};
\node[main node](4)[below right of=3]{};
\node[main node](5)[above right of=4]{};
\node[main node](6)[below right of=5]{};
\node[main node](7)[above right of=6]{};
\node[main node](8)[below right of=7]{};

\path[every node/.style={font=\sffamily\footnotesize}, thin]
(2) edge [bend left=20] (3)
	edge (3)
	edge [bend right=20] (3)
	edge [bend right=10] (4)

(3) edge [bend left=20] (4)
	edge (4)
	edge [bend right=20] (4)
	edge [bend left=10] (5)

(4) edge [bend left=20] (5)
	edge (5)
	edge [bend right=20] (5)
	edge [bend right=10] (6)

(5) edge [bend left=20] (6)
	edge (6)
	edge [bend right=20] (6)
	edge [bend left=10] (7)

(6) edge [bend left=20] (7)
	edge (7)
	edge [bend right=20] (7)
	edge [bend right=10] (8)	
	
(7) edge [bend left=20] (8)
	edge (8)
	edge [bend right=20] (8);
	
\end{tikzpicture}
\caption{The starting graph $\mathcal{G}_{a..b}^{\ec}$ \label{fig:exemple_eqManqute1}}
\begin{tikzpicture}[
  >=stealth',
  shorten >=0.8pt,
  auto,
  node distance=0.8cm,
  thick,
  main node/.style={circle, thin,fill=blue!20,draw,font=\sffamily\bfseries},
  second node/.style={circle, thin,fill=red!20,draw,font=\sffamily\bfseries}
  ]

\node[main node](2){};
\node[main node](3)[above right of=2]{};
\node[main node](4)[below right of=3]{};
\node[main node](5)[above right of=4]{};
\node[main node](6)[below right of=5]{};
\node[main node](7)[above right of=6]{};
\node[main node](8)[below right of=7]{};
\node[second node](9)[above right of=8]{};

\path[every node/.style={font=\sffamily\footnotesize}, thin]
(2) edge [bend left=20] (3)
	edge (3)
	edge [bend right=20] (3)
	edge [bend right=10] (4)

(3) edge [bend left=20] (4)
	edge (4)
	edge [bend right=20] (4)
	edge [bend left=10] (5)

(4) edge [bend left=20] (5)
	edge (5)
	edge [bend right=20] (5)
	edge [bend right=10] (6)

(5) edge [bend left=20] (6)
	edge (6)
	edge [bend right=20] (6)
	edge [bend left=10] (7)

(6) edge [bend left=20] (7)
	edge (7)
	edge [bend right=20] (7)
	edge [bend right=10] (8)	
	
(7) edge [bend left=20] (8)
	edge (8)
	edge [bend right=20] (8)
	edge [bend left=10, color=red] (9)	

(8) edge [bend left=20, color=red] (9)
	edge [color = red](9)
	edge [bend right=20, color=red] (9);
	
\end{tikzpicture}
\caption{$\G_{a..b+1}^{\ec}$ contains a missing parity-check equation\label{fig:exemple_eqManqute2}}
\begin{tikzpicture}[
  >=stealth',
  shorten >=0.8pt,
  auto,
  node distance=0.8cm,
  thick,
  main node/.style={circle, thin,fill=blue!20,draw,font=\sffamily\bfseries},
  second node/.style={circle, thin,fill=red!20,draw,font=\sffamily\bfseries}
  ]

\node[main node](2){};
\node[main node](3)[above right of=2]{};
\node[main node](4)[below right of=3]{};
\node[main node](5)[above right of=4]{};
\node[main node](6)[below right of=5]{};
\node[main node](7)[above right of=6]{};
\node[main node](8)[below right of=7]{};
\node[second node](9)[above right of=8]{};
\node[main node](10)[below right of=9]{};
\node[main node](11)[above right of=10]{};
\node[main node](12)[below right of=11]{};
\node[main node](13)[above right of=12]{};

\path[every node/.style={font=\sffamily\footnotesize}, thin]
(2) edge [bend left=20] (3)
	edge (3)
	edge [bend right=20] (3)
	edge [bend right=10] (4)

(3) edge [bend left=20] (4)
	edge (4)
	edge [bend right=20] (4)
	edge [bend left=10] (5)

(4) edge [bend left=20] (5)
	edge (5)
	edge [bend right=20] (5)
	edge [bend right=10] (6)

(5) edge [bend left=20] (6)
	edge (6)
	edge [bend right=20] (6)
	edge [bend left=10] (7)

(6) edge [bend left=20] (7)
	edge (7)
	edge [bend right=20] (7)
	edge [bend right=10] (8)	
	
(7) edge [bend left=20] (8)
	edge (8)
	edge [bend right=20] (8)
	edge [bend left=10, color=red] (9)	

(8) edge [bend left=20, color=red] (9)
	edge [color = red](9)
	edge [bend right=20, color=red] (9)
	edge [bend right=10] (10)	

(9) edge [bend left=20, color=red] (10)
	edge [color = red](10)
	edge [bend right=20, color=red] (10)
	edge [bend left=10, color=red] (11)	

(10) edge [bend left=20] (11)
	edge (11)
	edge [bend right=20] (11)
	edge [bend right=10] (12)

(11) edge [bend left=20] (12)
	edge (12)
	edge [bend right=20] (12)
	edge [bend left=10] (13)

(12) edge [bend left=20] (13)
	edge (13)
	edge [bend right=20] (13);
	
\end{tikzpicture}
\caption{Continue to extend the graph\label{fig:exemple_eqManqute3}}
\end{center}
\end{figure}

At the end we recover the isomorphism between $\G^\ec$ and $\G^\pi$. Indeed, at the end of the extension, $\G^\ec_{a..b}$ is equal to $\G^\ec$ and $\G^\pi_{a..b}$ to $\G^\pi$ probably with additional vertices representing missing parity-check equations.

\subsection{Reconstructing the interleaver}
Now we have the isomorphism between $\G^{\ec_\CC}$ and $\G^\pi$, so to reconstruct the interleaver we need to recover the isomorphism $\psi$ between labels on edges of $\Gtilde^{\ec_\CC}$ and $\Gtilde^\pi$.\\

We recover $\psi$ step by step, at each step we search for a sub-graph of $\Gtilde^{\ec_\CC}$ which has a label $i$ appearing only once, or appearing a different number of times than the other labels. The label of the image by $\phi$ of this edge labeled $i$ gives us $\psi(i)$. Then we remove all edges labeled by $i$ in  $\G^{\ec_\CC}$ and by $\psi(i)$ in $\Gtilde^{\ec_\pi}$.\\

At the end of this extension, we extend $\psi$ with positions which do not appear on graphs but are involved in parity-check equations represented by these graphs.\\

$\psi$ defines the interleaver $\pi$, indeed $\pi(i)= \psi(i+m^{\ec_\CC})$ where $m^{\ec_\CC}$ is the minimal value such that $\G^{\ec_\CC}$ contains a vertex representing a parity-check equation involving the position $m^{\ec_\CC}$.\\

\begin{remark}[Several isomorphisms]
Depending on $\ec_\mathcal{C}$ there might be several bijections between labels of the two graphs. In this case we have several interleavers. For these interleavers only the first and last positions are different. The number of interleavers just depends on $\ec_\mathcal{C}$ and not on the length of $\pi$.\\
\end{remark}

\begin{example}
\label{ex:bijection}
The size of the interleaver is $N = 26$. The two graphs $\Gtilde^{\ec_\CC}$ and $\Gtilde^{\pi}$ are represented on Figures \ref{fig:exemple_permutation1} and \ref{fig:exemple_permutation2}. $\phi$ is defined by $\phi : \Gtilde^{\ec_\CC} \rightarrow \Gtilde^{\pi}$, $\ec_\CC^{(-4)} \mapsto \ec_5$, $\ec_\CC^{(-3)} \mapsto \ec_6$, $\ec_\CC^{(-2)} \mapsto \ec_3$, \dots, $\ec_\CC^{(6)}\mapsto \ec_2$. \\
If we take the sub-graph of $\Gtilde^{\ec_\CC}$ induced by $\ec_\CC^{(-1)}$, $\ec_\CC^{(0)}$ and $\ec_\CC^{(1)}$, then we deduce that $\psi(3)=1$ because the label $3$ appears tree times and no other label appears tree times in this sub-graph. Then we also deduce that $\psi(1)=3$, $\psi(4)=25$ and $\psi(5)=17$. With other sub-graphs we obtain the isomorphism $\psi$ defined by $\psi :\Gtilde^{\ec_\CC} \rightarrow \Gtilde^{\pi}$ : \\
\begin{minipage}{0.16 \linewidth}
$-5 \mapsto 26$,\\
$-3 \mapsto 12$,\\
$-2 \mapsto 8$,\\
$-1 \mapsto 5$,\\
\end{minipage}
\begin{minipage}{0.16 \linewidth}
$0 \mapsto 20$,\\
$1 \mapsto 3$, \\
 $2 \mapsto 15$,\\
 $3 \mapsto 1$,\\
\end{minipage}
\begin{minipage}{0.16 \linewidth}
 $4 \mapsto 25$,\\
 $5 \mapsto 17$,\\
 $6 \mapsto 23$,\\
 $7 \mapsto 6$,\\
\end{minipage}
\begin{minipage}{0.16 \linewidth}
  $8 \mapsto 13$,\\
 $9 \mapsto 11$, \\
 $10 \mapsto 7$, \\
 $11 \mapsto 16$, \\
\end{minipage}
\begin{minipage}{0.16 \linewidth}
 $12 \mapsto 19$,\\
 $13 \mapsto 2$,\\
 $14 \mapsto 21$,\\
 $15 \mapsto 10$. \\
\end{minipage}

With this bijection we deduce a part of the interleaver $\pi$ : \\
$\pi = [...,26,?,12,8,5,20,3,15,1,25,17,23,6,13,11,7,16,19,2,21,10,...]$\\
On the graphs we do not represent the positions involved in a single parity-check equation, we do not have edges with the corresponding label. But we know these values and we use them to determine the first and last positions of $\pi$. For example, the second parity-check equation represented in $\Gtilde^{\ec_\mathcal{C}}$ is $\ec_\CC^{(-3)}=\{-5,-4,-3,-1,0\}$ and the image by $\phi$ of this parity-check equation is $\ec_6=\{5,12,18,20,26\}$. With this parity-check equation, we extend the bijection with $-4 \mapsto 18$ (the only two unused values in these parity-check equations). With the same reasoning we deduce that $16 \mapsto 4$. The first parity-check equation represented on $\Gtilde^{\ec_\mathcal{C}}$ is $\ec_\CC^{(-4)}=\{-7,-6,-5,-3,-2\}$ and the corresponding equation on  $\Gtilde^\pi$  is $\ec_5=\{8,9,12,14,26\}$, with these parity-check equations we deduce that $-7 \mapsto 9$ and $-6 \mapsto 14$ or $-7 \mapsto 14$ and $-6 \mapsto 9$, we have the same indeterminate for the two last positions. So we have 4 possible interleavers:\\ 
$\pi = [14,9,26,18,12,8,5,20,3,15,1,25,17,23,6,13,11,7,16,19,2,21,10,4,22,24]$\\
or $\pi = [9,14,26,18,12,8,5,20,3,15,1,25,17,23,6,13,11,7,16,19,2,21,10,4,22,24]$\\
or $\pi = [14,9,26,18,12,8,5,20,3,15,1,25,17,23,6,13,11,7,16,19,2,21,10,4,24,22]$\\
or $\pi = [9,14,26,18,12,8,5,20,3,15,1,25,17,23,6,13,11,7,16,19,2,21,10,4,24,22]$
Moreover we can also take the mirror of  the isomorphism $\phi$, and we obtain 4 new possible interleavers. \\

\end{example}

\begin{figure}
\begin{tikzpicture}[
  >=stealth',
  shorten >=2pt,
  auto,
  node distance=2cm,
  thick,
  main node/.style={circle, thin,fill=blue!20,draw,font=\sffamily\bfseries},
  second node/.style={circle, thin,fill=red!20,draw,font=\sffamily\bfseries}
  ]

\node[main node](2){\tiny $\ec_\CC^{(-4)}$};
\node[main node](3)[above right of=2]{\tiny $\ec_\CC^{(-3)}$};
\node[main node](4)[below right of=3]{\tiny $\ec_\CC^{(-2)}$};
\node[main node](5)[above right of=4]{\tiny $\ec_\CC^{(-1)}$};
\node[main node](6)[below right of=5]{\tiny $\ec_\CC^{(0)}$};
\node[main node](7)[above right of=6]{\tiny $\ec_\CC^{(1)}$};
\node[main node](8)[below right of=7]{\tiny $\ec_\CC^{(2)}$};
\node[main node](9)[above right of=8]{\tiny $\ec_\CC^{(3)}$};
\node[main node](10)[below right of=9]{\tiny $\ec_\CC^{(4)}$};
\node[main node](11)[above right of=10]{\tiny $\ec_\CC^{(5)}$};
\node[main node](12)[below right of=11]{\tiny $\ec_\CC^{(6)}$};

\path[every node/.style={font=\sffamily\footnotesize}, thin]

(2) edge [bend left=15] node[right]{$-5$} (3)
	edge [bend right=15] node[left]{$-3$}(3)
	edge [bend right=10] node[below]{$-3$}(4)
	edge [bend left=20] node[below]{$-2$}(4)
	
(3) edge [bend left=15] node[left]{$-3$} (4)
	edge [bend right=15] node[right]{$-1$}(4)
	edge [bend right=10] node[above]{$-1$}(5)
	edge [bend left=20] node[above]{$0$}(5)
	
(4) edge [bend left=15] node[right]{$-1$} (5)
	edge [bend right=15] node[left]{$1$}(5)
	edge [bend right=10] node[below]{$1$}(6)
	edge [bend left=20] node[below]{$2$}(6)
	
(5) edge [bend left=15] node[left]{$1$} (6)
	edge [bend right=15] node[right]{$3$}(6)
	edge [bend right=10] node[above]{$3$}(7)
	edge [bend left=20] node[above]{$4$}(7)
	
(6) edge [bend left=15] node[right]{$3$} (7)
	edge [bend right=15] node[left]{$5$}(7)
	edge [bend right=10] node[below]{$5$}(8)
	edge [bend left=20] node[below]{$6$}(8)

(7) edge [bend left=15] node[left]{$5$} (8)
	edge [bend right=15] node[right]{$7$}(8)
	edge [bend right=10] node[above]{$7$}(9)
	edge [bend left=20] node[above]{$8$}(9)

(8) edge [bend left=15] node[right]{$7$} (9)
	edge [bend right=15] node[left]{$9$}(9)
	edge [bend right=10] node[below]{$9$}(10)
	edge [bend left=20] node[below]{$10$}(10)	

(9) edge [bend left=15] node[left]{$9$} (10)
	edge [bend right=15] node[right]{$11$}(10)
	edge [bend right=10] node[above]{$11$}(11)
	edge [bend left=20] node[above]{$12$}(11)
	
(10) edge [bend left=15] node[right]{$11$} (11)
	edge [bend right=15] node[left]{$13$}(11)
	edge [bend right=10] node[below]{$13$}(12)
	edge [bend left=20] node[below]{$14$}(12)

(11) edge [bend left=15] node[left]{$13$} (12)
	edge [bend right=15] node[right]{$15$}(12);
	
\end{tikzpicture}
\caption{The graph $\Gtilde_2^{\ec_\mathcal{C}}$ \label{fig:exemple_permutation1}}
\begin{tikzpicture}[
  >=stealth',
  shorten >=2pt,
  auto,
  node distance=2cm,
  thick,
  main node/.style={circle, thin,fill=blue!20,draw,font=\sffamily\bfseries},
  second node/.style={circle, thin,fill=red!20,draw,font=\sffamily\bfseries}
  ]

\node[main node](2){\small $\ec_5$};
\node[main node](3)[above right of=2]{\small $\ec_6$};
\node[main node](4)[below right of=3]{\small $\ec_3$};
\node[main node](5)[above right of=4]{\small $\ec_8$};
\node[main node](6)[below right of=5]{\small $\ec_4$};
\node[main node](7)[above right of=6]{\small $\ec_9$};
\node[main node](8)[below right of=7]{\small $\ec_{11}$};
\node[main node](9)[above right of=8]{\small $\ec_1$};
\node[main node](10)[below right of=9]{\small $\ec_7$};
\node[main node](11)[above right of=10]{\small $\ec_{10}$};
\node[main node](12)[below right of=11]{\small $\ec_2$};

\path[every node/.style={font=\sffamily\footnotesize}, thin]

(2) edge [bend left=15] node[right]{$26$}(3)
	edge [bend right=15] node[left]{$12$}(3)
	edge [bend right=10] node[below]{$12$}(4)
	edge [bend left=20] node[below]{$8$}(4)
	
(3) edge [bend left=15] node[left]{$12$} (4)
	edge [bend right=15] node[right]{$5$}(4)
	edge [bend right=10] node[above]{$5$}(5)
	edge [bend left=20] node[above]{$20$}(5)
	
(4) edge [bend left=15] node[right]{$5$} (5)
	edge [bend right=15] node[left]{$3$}(5)
	edge [bend right=10] node[below]{$3$}(6)
	edge [bend left=20] node[below]{$15$}(6)
	
(5) edge [bend left=15] node[left]{$3$} (6)
	edge [bend right=15] node[right]{$1$}(6)
	edge [bend right=10] node[above]{$1$}(7)
	edge [bend left=20] node[above]{$25$}(7)
	
(6) edge [bend left=15] node[right]{$1$} (7)
	edge [bend right=15] node[left]{$17$}(7)
	edge [bend right=10] node[below]{$17$}(8)
	edge [bend left=20] node[below]{$23$}(8)

(7) edge [bend left=15] node[left]{$17$} (8)
	edge [bend right=15] node[right]{$6$}(8)
	edge [bend right=10] node[above]{$6$}(9)
	edge [bend left=20] node[above]{$13$}(9)

(8) edge [bend left=15] node[right]{$6$} (9)
	edge [bend right=15] node[left]{$11$}(9)
	edge [bend right=10] node[below]{$11$}(10)
	edge [bend left=20] node[below]{$7$}(10)	

(9) edge [bend left=15] node[left]{$11$} (10)
	edge [bend right=15] node[right]{$16$}(10)
	edge [bend right=10] node[above]{$16$}(11)
	edge [bend left=20] node[above]{$19$}(11)
	
(10) edge [bend left=15] node[right]{$16$} (11)
	edge [bend right=15] node[left]{$2$}(11)
	edge [bend right=10] node[below]{$2$}(12)
	edge [bend left=20] node[below]{$21$}(12)

(11) edge [bend left=15] node[left]{$2$} (12)
	edge [bend right=15] node[right]{$10$}(12);
	
\end{tikzpicture}
\caption{The graph $\Gtilde_2^\ec$ \label{fig:exemple_permutation2}}
\end{figure}

\subsection{Particular cases}
\textit{Indeterminate positions.} If the reconstructed interleaver contains indeterminate positions, we search for these positions using noisy interleaved codewords. At each indeterminate positions we test all possible values. To test a position we reconstruct the missing parity-check equations and we verify the number of noisy interleaved codewords that satisfy these parity-check equations. If this number is less than a threshold (we can take $\frac{1}{2}\frac{1-(1-2\tau)^t}{2}M$) it is not the correct value for this position. 

\begin{example}
The graph associated to $\Gtilde_2^{\ec_\mathcal{C}}$ is on Figure \ref{fig:exemple_permutation1}, and the graph $\Gtilde_2^\pi$ on Figure \ref{fig:exemple_permutation3}. One bijection on edges $\psi$ is the same as in Example \ref{ex:bijection}, but we can not know if $7 \mapsto 6$ and $8 \mapsto 13$ or $7 \mapsto 13$ and $8 \mapsto 6$. To determine the right bijection we reconstruct the missing parity-check equation and we test them with noisy interleaved codewords. We test $7 \mapsto 6$ and $8 \mapsto 13$, in this case the missing equation is $\{6,7,11,17,23\}$, then we test $7 \mapsto 13$ and $8 \mapsto 6$, the missing equation is $\{7,11,13,17,23\}$. With the number of noisy interleaved codewords that satisfy these equations we deduce the 8 possible interleavers as in Example \ref{ex:bijection}.

\end{example}

\begin{figure}
\begin{tikzpicture}[
  >=stealth',
  shorten >=2pt,
  auto,
  node distance=2cm,
  thick,
  main node/.style={circle, thin,fill=blue!20,draw,font=\sffamily\bfseries},
  second node/.style={circle, thin,fill=red!20,draw,font=\sffamily\bfseries}
  ]

\node[main node](2){\small $\ec_5$};
\node[main node](3)[above right of=2]{\small $\ec_6$};
\node[main node](4)[below right of=3]{\small $\ec_3$};
\node[main node](5)[above right of=4]{\small $\ec_8$};
\node[main node](6)[below right of=5]{\small $\ec_4$};
\node[main node](7)[above right of=6]{\small $\ec_9$};
\node[second node](8)[below right of=7]{};
\node[main node](9)[above right of=8]{\small $\ec_1$};
\node[main node](10)[below right of=9]{\small $\ec_7$};
\node[main node](11)[above right of=10]{\small $\ec_{10}$};
\node[main node](12)[below right of=11]{\small $\ec_2$};

\path[every node/.style={font=\sffamily\footnotesize}, thin]

(2) edge [bend left=15] node[right]{$26$}(3)
	edge [bend right=15] node[left]{$12$}(3)
	edge [bend right=10] node[below]{$12$}(4)
	edge [bend left=20] node[below]{$8$}(4)
	
(3) edge [bend left=15] node[left]{$12$} (4)
	edge [bend right=15] node[right]{$5$}(4)
	edge [bend right=10] node[above]{$5$}(5)
	edge [bend left=20] node[above]{$20$}(5)
	
(4) edge [bend left=15] node[right]{$5$} (5)
	edge [bend right=15] node[left]{$3$}(5)
	edge [bend right=10] node[below]{$3$}(6)
	edge [bend left=20] node[below]{$15$}(6)
	
(5) edge [bend left=15] node[left]{$3$} (6)
	edge [bend right=15] node[right]{$1$}(6)
	edge [bend right=10] node[above]{$1$}(7)
	edge [bend left=20] node[above]{$25$}(7)
	
(6) edge [bend left=15] node[right]{$1$} (7)
	edge [bend right=15] node[left]{$17$}(7)
	edge [bend right=10,color=red](8)
	edge [bend left=20,color=red](8)

(7) edge [bend left=15,color=red](8)
	edge [bend right=15,color=red](8)
	edge [bend right=10] node[above]{$6$}(9)
	edge [bend left=20] node[above]{$13$}(9)

(8) edge [bend left=15,color=red](9)
	edge [bend right=15,color=red](9)
	edge [bend right=10,color=red](10)
	edge [bend left=20,color=red](10)	

(9) edge [bend left=15] node[left]{$11$} (10)
	edge [bend right=15] node[right]{$16$}(10)
	edge [bend right=10] node[above]{$16$}(11)
	edge [bend left=20] node[above]{$19$}(11)
	
(10) edge [bend left=15] node[right]{$16$} (11)
	edge [bend right=15] node[left]{$2$}(11)
	edge [bend right=10] node[below]{$2$}(12)
	edge [bend left=20] node[below]{$21$}(12)

(11) edge [bend left=15] node[left]{$2$} (12)
	edge [bend right=15] node[right]{$10$}(12);
	
\end{tikzpicture}
\caption{The graph $\Gtilde_2^\ec$ \label{fig:exemple_permutation3}}
\end{figure}

\textit{Not the right length.} If the reconstructed interleaver has not the right length, it is the case when $\mathcal{L}_1$ does not contain all parity-check equations of the same type, the missing parity-check equations are not classified. To recover the beginning and the end of the interleaver we continue the reconstruction by applying the same steps using unclassified parity-check equations: we extend the graphs $\mathcal{G}_{a..b}^\pi$ and $\mathcal{G}_{a..b}^{\ec_\mathcal{C}}$ then we label these graphs and deduce the entire interleaver.

\section{Experimental results}
We have run several experimental tests for different convolutional codes $\mathcal{C}$ and interleaver sizes $N$.\\
In the first test we used the convolutional code defined by the generator matrix in polynomial form $\mathcal{C}^1 = (1+D+D^2+D^5, 1+D+D^3+D^4+D^6)$. This code satisfies one parity-check equation of weight $8$. With a set of interleaved codewords we search for parity-check equations of weight $8$ of $\mathcal{C}^1_\pi$ using a slightly improved  method of \cite{CF09} (we give in Table \ref{table:nombreMotsEtTempsEq} the number $M$ of codewords that we use and the running time for recovering all parity-check equations), then we applied our method to reconstruct the interleaver and the convolutional code. For these tests we assumed that $s_{max} = 25$, and we give the running time in Table \ref{table:resCodePoids8}.\\

In the next test, the convolutional code was ${\mathcal{C}^2 =(1+D+D^2, 1+D^2+D^3)}$. This code has 5 types of parity-check equation of weight $6$. To test our method with $\mathcal{C}^2$ we assumed that $s_{max}=10$, the running times are also in Table \ref{table:resCodePoids8} and \ref{table:nombreMotsEtTempsEq}.\\

The last test was with the the convolutional code defined by $\mathcal{C}^3 =(1+D^2+D^3+D^5+D^6, 1+D+D^2+D^3+D^6)$. This code satisfies $11$ types of parity-check equations of weight $10$. To reconstruct the interleaver and the convolutional code we assumed $s_{max}=20$, see Table \ref{table:resCodePoids8} and \ref{table:nombreMotsEtTempsEq}.\\

\begin{table}
\begin{center}
\begin{tabular}{c||p{0.5cm}p{1.5cm}p{1.5cm}p{1.5cm}}
$N$ & \multicolumn{4}{c}{running time (in seconds)}\\
 && $\mathcal{C}^1$  & $\mathcal{C}^2$  & $\mathcal{C}^3$ \\
 \hline \hline
$1\,000$ & &5 & 0.2 & 5\\
$2\,000$ & &6 & 0.7 & 10\\
$5\,000$ & &7 & 4 & 60\\
$8\,000$ & &11 & 10 & 130\\
$10\,000$ & &12 & 15 & 185
\end{tabular}
\caption{Running time for ${\mathcal{C}^1 =(1+D+D^2+D^5, 1+D+D^3+D^4+D^6)}$, ${\mathcal{C}^2 =(1+D+D^2, 1+D^2+D^3)}$ and $\mathcal{C}^3 =(1+D^2+D^3+D^5+D^6, 1+D+D^2+D^3+D^6)$\label{table:resCodePoids8}}
\end{center}
\end{table}

\begin{table}
\begin{center}
\begin{tabular}{c||p{1cm}p{1cm}|p{1cm}p{1cm}|p{1cm}p{1cm}}
  & \multicolumn{2}{c}{$\mathcal{C}^1$}  & \multicolumn{2}{c}{$\mathcal{C}^2$}  & \multicolumn{2}{c}{$\mathcal{C}^3$} \\
  $N$ & $M$ & run. time & $M$ & run. time & $M$ & run. time\\
 \hline \hline
$1\,000$ & 400 & 60 & 200 & 3 & 400 & 300\\
$2\,000$ & 500 & 60 & 400 & 8 & 600 & 600 \\
$5\,000$ & 1400 & 600 & 900 & 45 & 1600 & 2000\\
$8\,000$ & 2000 & 1800 & 1300 & 240 & 2400 & 3000\\
$10\,000$ & 2600 & 2700 & 1700 & 300 & 2800 & 9000\\
\end{tabular}
\caption{Running time in seconds for recovering all parity-check equations ${\mathcal{C}^1 =(1+D+D^2+D^5, 1+D+D^3+D^4+D^6)}$, ${\mathcal{C}^2 =(1+D+D^2, 1+D^2+D^3)}$ and $\mathcal{C}^3 =(1+D^2+D^3+D^5+D^6, 1+D+D^2+D^3+D^6)$ without noise \label{table:nombreMotsEtTempsEq}}
\end{center}
\end{table}

In all cases, the interleaver and the convolutional code were reconstructed efficiently. To obtain these running times we used all parity-check equations of weight 8, 6 or 10. 
Recovering all parity-check equations of low weight may take time,  but  our method can be applied without having all parity-check equations. For example, with the first convolutional code $\mathcal{C}^1 =(1+D+D^2+D^5, 1+D+D^3+D^4+D^6)$, we note in Table \ref{table:resPasToutesEquations} the running time for reconstructing the interleaver and the convolutional code in case we have less than $100\%$ of parity-check equations of weight $8$.\
We can see that, for small lengths the time increases rapidly if we do not have all parity-check equations, but for large lengths having all parity-check equations is not necessary to reconstruct to reconstruct efficiently the interleaver and the convolutional code.
\begin{table}
\begin{center}
\begin{tabular}{c||cc}
$N$ & $\%$ of parity-check & running time \\
 & equations  &   (in seconds)\\
 \hline \hline
$1\,000$ & 100 &5\\
& 99 & 7\\
& 96 & 37\\
& 93 & 110\\
\hline
$2\,000$ & 100 &6 \\
& 95 & 13\\
& 93 & 155\\
\hline
$5\,000$ & 100 &7 \\
& 95 & 78\\
\hline
$8\,000$ &100  &11\\
& 97 & 21\\
\hline
$10\,000$ & 100 & 12 \\
& 94 & 63\\
\end{tabular}
\caption{Running time for ${\mathcal{C}^1 =(1+D+D^2+D^5, 1+D+D^3+D^4+D^6)}$\label{table:resPasToutesEquations}}
\end{center}
\end{table}

We also test with noisy interleaver codewords, for the convolutional code $\CC^2$ and the binary symmetric channel of crossover probability $p=0.001$ and $p=0.01$, we note in table \ref{table:nombreMotsEtTempsEqBruit} the running time to recover almost all parity-check equations (more than $96\%$ of them). These times are long but by parallelizing, the running time is divided by as much as executed programs. The running time to reconstruct the convolutional code and the interleaver is the same as in noiseless case.

\begin{table}
\begin{center}
\begin{tabular}{c||p{1cm}p{1.5cm}|p{1cm}p{1.5cm}}
	& \multicolumn{2}{c|}{$p = 0.001$} & \multicolumn{2}{c}{$p = 0.01$}\\
  $N$ & $M$ & runn. time & $M$ & runn. time\\
 \hline \hline
$100$ & 100 & 1 & 100 & 10\\
$200$ & 100 & 3 & 100 & 240\\
$500$ & 300 & 30 & 200 & 4\,000\\
$1\,000$ & 400 & 360 & 200 & 72\,000\\
$2\,000$ & 600 & 16\,000\\
\end{tabular}
\caption{Running time in seconds for recovering all parity-check equations when ${\CC =(1+D+D^2, 1+D^2+D^3)}$ and with a crossover probability $p$ \label{table:nombreMotsEtTempsEqBruit}}
\end{center}
\end{table}

\section{Conclusion}
This paper shows that when an interleaved convolutional code is used, then it can be efficiently reconstructed  
from the knowledge of a few hundred (or thousand) observed noisy codewords in the case of moderate noise 
of the channel by first recovering low-weight codewords in the dual of the interleaved convolutional code and then using this set of dual codewords to recover the convolutional structure and the interleaver.
 This assumption of moderate noise can be removed when the length $N$ of the interleaver
is sufficiently short (say below a few hundred) and is needed to ensure that most low-weight 
codewords are obtained by the slightly improved Cluzeau-Finiasz method \cite{CF09} we used in our 
tests.
Once these parity-check equations are recovered, a graph representing how these parity-check equations intersect
is used to recover at the same time the interleaver and the convolutional code.
This method is really fast, for instance the second phase took less than a few minutes in all our experiments and this 
even for very long  interleavers (up to length $N=10000$).
This method applies to any convolutional code, it just needs convolutional codes that have reasonably low-weight and low-span codewords in the dual of the convolutional code, which 
is the case for virtually all convolutional codes used in practice.

\bibliographystyle{IEEEtran}
\bibliography{biblio}

\end{document}